\documentclass[11pt]{llncs}

\usepackage{etex}
\usepackage{graphicx}
\usepackage[disable,colorinlistoftodos]{todonotes}
\usepackage{authblk}
\usepackage{tabularx}
\usepackage{bussproofs}
\usepackage{proof}
\usepackage{xypic}
\usepackage{fancyhdr}
\usepackage{lscape}

\usepackage[english]{babel}
\usepackage{url}
\usepackage{algorithm}
\usepackage{algorithmic}
\usepackage{amsmath} 
\usepackage{amssymb}
\usepackage{color}

\newtheorem{fact}{Fact}

\newcommand{\imply}{\ensuremath{\rightarrow}}
\newcommand{\mil}{\ensuremath{\mathbf{M}_{\rightarrow}}}

\newcommand{\SEQ}{\ensuremath{\Rightarrow}}
\newcommand{\Seq}[2]{\ensuremath{#1 \SEQ #2}}

\def\fCenter{ \SEQ\ }

\title{How many times do we need an assumption ?}

\author{Edward Hermann Haeusler}

\institute{Dep. Inform\'atica \\ PUC-Rio \\ \texttt{hermann@inf.puc-rio.br}}

\date{}

\begin{document}

\maketitle

\begin{abstract}
In this article we present a class of formulas $\varphi_n$, $n\in Nat$, that need at least $2^{n}$ assumptions to be proved in a normal proof in Natural Deduction for purely implicational minimal propositional logic. In purely implicational classical propositional logic, with Peirce's rule, each $\varphi_n$ is proved with only one assumption in Natural Deduction in a normal proof. Hence, the formulas $\varphi_n$ have exponentially sized proofs in cut-free Sequent Calculus and  Tableaux. In fact $2^n$ is the lower-bound for normal proofs in ND, cut-free Sequent proofs and Tableaux. We discuss the consequences of the existence of this class of formulas for designing automatic proof-procedures based on these deductive systems.  
\end{abstract} 


\section{Introduction}\label{intro}

Providing proofs for propositional tautologies seems to be a hard task.  Huge proofs are such that their size is  super-polynomial with regard to  the size of their conclusions.  Knowing that there is a classical propositional logic tautology having only huge proofs is related to know whether $NP=CoNP$ or not (see \cite{Cook}). Intuitionistic logic is PSPACE-complete (\cite{Ladner}) and Richard Statman (see \cite{Statman}) showed that purely implicational minimal logic (\mil) is PSPACE-complete too. 
We showed in \cite{Haeusler} that, if a propositional logic has a Natural Deduction (ND) with the sub-formula property then it is PSPACE-Hard. This follows from the fact that \mil polynomially encodes any propositional logic that has such ND system.  Thus, 
the existence of huge proofs for a more general class of propositional logics is related to the existence of huge proofs in \mil that amounts to know whether $PSPACE=NP$ or not. The relations between these computational complexity classes and the existence of huge proofs involve arbitrary proof systems, indeed. For example, $NP=PSPACE$ is the case, if and only if, for any \mil tautology there is a proof system that produces a polynomially sized proof of this tautology. 

Dealing with arbitrary and general proof systems is quite hard, and is out of scope of this article. However, studying particular proof systems for key logics, like \mil or classical logic, can shed some light on practical aspects of implementing propositional theorem provers from the efficiency and economy of storage point of view.    

\mil carries almost all the proof-theoretical and logical information to produce polynomially bounded proofs in well-behaved\footnote{With sub-formula property} propositional logics. Thus we can conclude that focusing investigations on \mil is worth of noticing. 

There are many proof systems for \mil. The most well-known are structural/analytic  proof systems. Well-known systems are the Sequent Calculus (\cite{Gentzen}, Natural Deduction (\cite{Gentzen} and \cite{Prawitz}) and Tableaux (\cite{Beth,Smullyan}) based. These systems, mainly the first and the third kind, are quite good in providing means to produce proofs automatically. The backward chaining procedure, for example, if applied to a Sequent Calculus based proof system provides an automatic way to produce proofs. The problem with these proof procedures is when a decision on which rule to apply has to be made and how to deal with non-provable formulas when it is the case. With respect to this feature of dealing with invalid formulas, the  literature on both systems, Sequent Calculus and Tableaux, provides methods that either produce a proof or a counter-model. 

In \mil,  to provide a counter-model is so hard as to provide a proof, since it is a PSPACE-complete problem and the complexity class is deterministic. We know that \mil has finite model property and that the size of the counter-model is super-polynomially  upper bounded with respect to  the formula. It is interesting to investigate how this is related to the size of proofs in \mil, or at least to have a concrete evidence that huge proofs may be the case. Our intention is not only show huge proofs in \mil. To do that, we can use the polynomial translations reported in \cite{Statman} or \cite{Haeusler} to generate a formula from the Pigeon-Hole principle into \mil. We know, from \cite{Haken}, that this formula has only super-polynomially sized proofs in Resolution, and hence in cut-free Sequent Calculus and the same happens to the translation to \mil in cut-free Sequent Calculus and Natural Deduction. It is quite hard to detect from this approach to obtain huge proofs in \mil any precise reason for the super-size of the resulting proof. We believe that directly focusing on \mil is a more promising path, since \mil has less combinatorial alternatives, less logical constants, less alternative deductive system. The genesis of huge proofs in \mil is interesting and may shed some new light in propositional logic complexity. 

This article, motivated by the relationship between counter-model and proof construction in an integrated process,  improves a bit the knowledge on this subject by showing  a class of formulas in \mil that have proofs with at least exponentially-many assumptions, and hence are super-polynomially sized. The relation of these formulas with counter-model generation is discussed in section~\ref{counter-model}. In fact what is reported here was generated by the need of providing counter-models in a naive proof-procedure  for \mil based in Sequent Calculus as it is briefly discussed in section~\ref{counter-model}. In section~\ref{sec:formula} we introduce the class of formulas and in section~\ref{limits} we show  that  
they have exponentially sized normal proofs in the usual Natural Deduction for \mil. In the same section we also show that this is a lower bound in \mil. In classical propositional logic, these formulas have linear-sized proofs as it is shown in section~\ref{sec:formula}.

 All the formal propositional proofs/derivations in this article are presented in Prawitz-style Natural Deduction. The size of these normal proofs/derivations is polynomially simulated by cut-free Sequent Calculus and/or Tableaux. Thus, the lower bound shown here also applies to them.

\section{The purely implicational minimal logic}

The (purely) implicational minimal logic \mil is the fragment of minimal logic containing only the logical constant $\imply$. Its semantics is the intuitionistic Kripke semantics restricted to $\imply$ only. Given propositional language $\mathcal{L}$, a \mil model is a structure $\left< U,\preceq, \mathcal{V} \right>$, where $U$ is a non-empty set (worlds), $\preceq$ is a partial order relation on $U$ and $\mathcal{V}$ is a function from $U$ into the power set of $\mathcal{L}$, such that if $i,j\in U$ and $i\preceq j$ then $\mathcal{V}(i)\subseteq\mathcal{V}(j)$. Given a model, the satisfaction relationship $\models$ between worlds, in the model,  and formulas is defined as:
\begin{itemize}
\item $\left< U,\preceq, \mathcal{V} \right>\models_{i}p$, $p\in\mathcal{L}$, iff, $p\in\mathcal{V}(i)$
\item $\left< U,\preceq, \mathcal{V} \right>\models_{i}\alpha_1\imply\alpha_2$, iff, for every $j\in U$, such that $i\preceq j$, if $\left< U,\preceq, \mathcal{V} \right>\models_{j}\alpha_1$ then $\left< U,\preceq, \mathcal{V} \right>\models_{j}\alpha_2$.
\end{itemize}

Obs: In (full) minimal logic, $\bot$ has no special meaning, so there is no item declaring that $\left< U,\preceq, \mathcal{V} \right>\not\models_{i}\bot$. We remind that \mil does not have the $\bot$ in its language.  

As usual a formula $\alpha$ is valid in a model $\mathcal{M}$, namely $\mathcal{M}\models\alpha$, if and only if, it is satisfiable in every world $i$ of the model, namely $\forall i\in U \mathcal{M}\models_{i}\alpha$. A formula is a \mil tautology, if and only if, it is valid in every model. A formula is satisfiable in \mil if it is valid in a model $\mathcal{M}$ of \mil.   
The problem of knowing whether a formula is satisfiable or not is trivial in \mil. Every formula is satisfiable in the model 
$\left< \{\star\},\preceq, \mathcal{V} \right>$, where $\star$ is the only world, and $p\in\mathcal{V}(\star)$, for every $p$. Thus, $SAT$ is not an interesting problem in \mil. The same cannot be told about knowing whether a formula is a \mil tautology or not. 

It is known that Prawitz Natural Deduction system for minimal logic with only the $\imply$-rules ($\imply$-Elim and $\imply$-Intro below)  is sound and complete for the \mil Krikpe semantics. As a consequence of this, Gentzen's $LJ$ system (see \cite{Takeuti}) containing only right and left $\imply$-rules is also sound and complete. As it is well-known one of these rules is not double-sounded  and not invertible. A naive proof-procedure for \mil based only on this usual Gentzen sequent calculus is not possible. 
\begin{prooftree} 
\AxiomC{$[\alpha]$}
\noLine
\UnaryInfC{$\mid$}
\noLine
\UnaryInfC{$\beta$}
\RightLabel{$\imply$-Intro}
\UnaryInfC{$\alpha\imply\beta$}
\AxiomC{$\alpha$}
\AxiomC{$\alpha\imply\beta$}
\RightLabel{$\imply$-Elim}
\BinaryInfC{$\beta$}
\noLine
\BinaryInfC{}
\end{prooftree}

In section~\ref{counter-model} we discuss the consequences of the existence of the class of formulas presented in this article to obtain a naive, complete and sound proof-procedure for \mil.  

\section{Needing exponentially many assumptions}\label{sec:formula}

In \cite{GillesJiang} we can find a  discussion on the fact that when proving theorems in a logic weaker than classical logic, the need of using an assumption more than once has a strong influence on how complex is the proof procedure and consequently the decision procedure for this logic. There, we can find the formula $((((A\imply B)\imply A)\imply A)\imply B)\imply B$. Considering the proof systems of ND and CS  mentioned in the previous section, this formula needs to use the assumption $((A\imply B)\imply A)\imply A)\imply B$ at least twice in order to be proved in \mil. Inspired by this example, we can define a class of formulas with no bounds on the use of assumptions. This shows that limiting the use of assumptions in an automatic proof-procedure for \mil is not an alternative that ensures completeness. In the sequel we define the class of formulas. Below you find a normal proof of 
$((((A\imply B)\imply A)\imply A)\imply B)\imply B$. Note that it cannot be proved with less than 2 assumptions of $(((A\imply B)\imply A)\imply A)\imply B$.

The following formula combines two instances of the formula mentioned above in order to have a formula that needs 4 times an assumption. 
\begin{eqnarray}
((((A\imply \xi)\imply A)\imply A)\imply \xi)\imply C\label{formula}
\end{eqnarray}
, where 
$\xi=(((D\imply C)\imply D)\imply D)\imply C$.

In figure~\ref{prova} we show a normal derivation of this formula~\ref{formula} above. We can see that it has 4 assumptions of 
$((A\imply \xi)\imply A)\imply A)\imply \xi)$. We can see how to use this pattern such that if it is repeated n-times we define a formula $\varphi_n$, such that, any normal proof of $\varphi_n$ has to use an assumption at least $2^n$ times, see section ~\ref{limits}. Before we proceed with $\varphi_n$ definition, we have to show that the need for repeating assumptions is not the case for classical propositional logic. 

Consider now that the logic is the purely implicational {\em classical} logic instead of the purely implicational minimal logic. That is, we consider the $\imply$ introduction and elimination rules, plus the classical absurdity rule, or the Peirce's rule: from $C\imply D\vdash C$ then infer $C$. Taking into account the version with Peirce's rule, we provide the proof of the formula~\ref{formula} with only use of assumption, as shown in figure~\ref{outraprova}. This comes from the fact that $(((D\imply C)\imply D)\imply D)$ in an instance of the implicational form of Peirce's rule, so it is provable. From this proof and $\xi=(((D\imply C)\imply D)\imply D)\imply C$ we prove $C$. $\xi$ itself is provable by means of a proof of the Peirce's formula $((A\imply \xi)\imply A)\imply A)$ and the $(((A\imply \xi)\imply A)\imply A)\imply \xi$ discharged to proof the desired formula. The purely implicational classical logic is not the focus of this article, in \cite{Studia} and \cite{LeoGordeev} we can find a detailed presentation of the purely implicational classical logic with some  proof-theoretic results. Our discussion on the classical setting has the purpose of showing how the use of classical logic can, in some cases, turns proofs smaller.

\begin{landscape}
\begin{figure}[h]
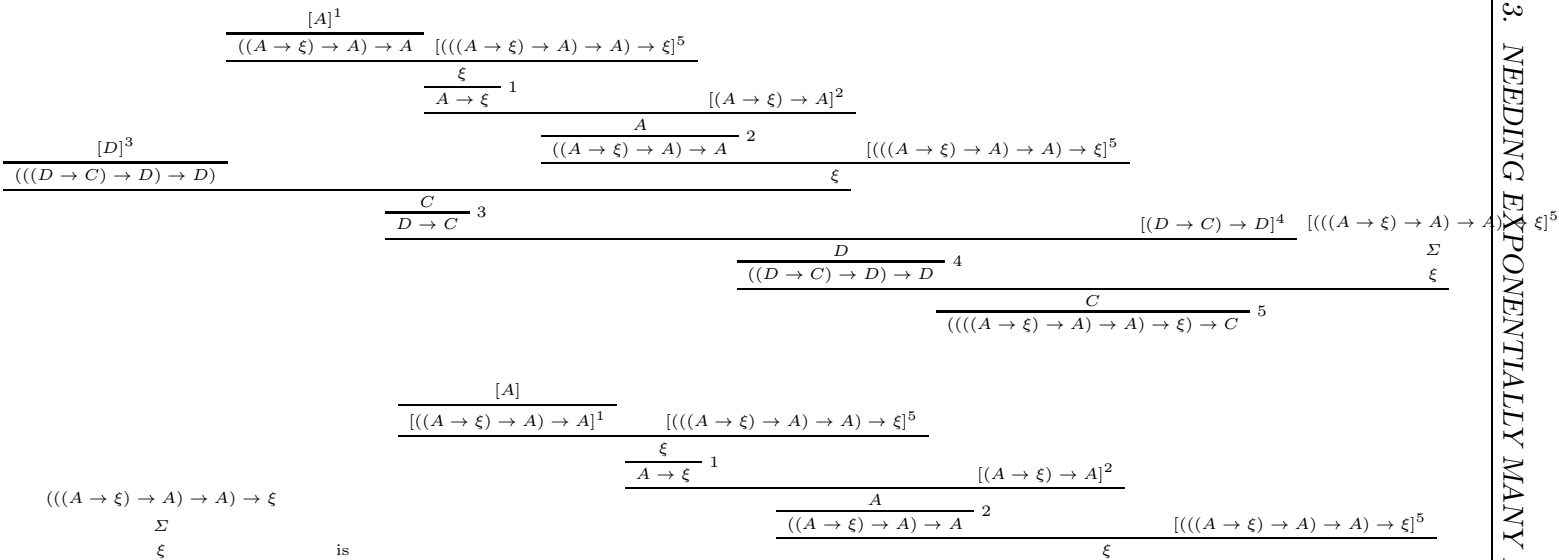

{\tiny
\begin{prooftree}
\def\defaultHypSeparation{\hskip 0mm}
\insertBetweenHyps{\hskip -3cm}
\AxiomC{$[D]^{3}$}
\UnaryInfC{$(((D\imply C)\imply D)\imply D)$}
\AxiomC{$[A]^{1}$}
\UnaryInfC{$((A\imply \xi)\imply A)\imply A$}
\AxiomC{$[(((A\imply \xi)\imply A)\imply A)\imply \xi]^{5}$}
\BinaryInfC{$\xi$}
\RightLabel{1}
\UnaryInfC{$A\imply\xi$}
\AxiomC{$[(A\imply\xi)\imply A]^{2}$}
\BinaryInfC{$A$}
\RightLabel{2}
\UnaryInfC{$((A\imply\xi)\imply A)\imply A$}
\AxiomC{$[(((A\imply \xi)\imply A)\imply A)\imply \xi]^{5}$}
\BinaryInfC{$\xi$}
\BinaryInfC{$C$}
\RightLabel{$3$}
\UnaryInfC{$D\imply C$}
\AxiomC{$[(D\imply C)\imply D]^{4}$}
\BinaryInfC{$D$}
\RightLabel{4}
\UnaryInfC{$((D\imply C)\imply D)\imply D$}
\AxiomC{$[(((A\imply \xi)\imply A)\imply A)\imply \xi]^{5}$}
\noLine
\UnaryInfC{$\Sigma$}
\noLine
\UnaryInfC{$\xi$}
\BinaryInfC{$C$}
\RightLabel{5}
\UnaryInfC{$((((A\imply \xi)\imply A)\imply A)\imply \xi)\imply C$}
\noLine
\UnaryInfC{}
\end{prooftree}
\begin{prooftree}
\AxiomC{$(((A\imply \xi)\imply A)\imply A)\imply \xi$}
\noLine
\UnaryInfC{$\Sigma$}
\noLine
\UnaryInfC{$\xi$}
\AxiomC{is}
\AxiomC{$[A]$}
\UnaryInfC{$[((A\imply \xi)\imply A)\imply A]^{1}$}
\AxiomC{$[(((A\imply \xi)\imply A)\imply A)\imply \xi]^{5}$}
\BinaryInfC{$\xi$}
\RightLabel{1}
\UnaryInfC{$A\imply\xi$}
\AxiomC{$[(A\imply\xi)\imply A]^{2}$}
\BinaryInfC{$A$}
\RightLabel{2}
\UnaryInfC{$((A\imply\xi)\imply A)\imply A$}
\AxiomC{$[(((A\imply \xi)\imply A)\imply A)\imply \xi]^{5}$}
\BinaryInfC{$\xi$}
\noLine
\TrinaryInfC{}
\end{prooftree}
}
\caption{Proof of the formula in purely implicational minimal logic}\label{prova}
\end{figure}
\end{landscape}

\begin{figure}[h]
{\tiny
\begin{prooftree}
\AxiomC{$\Pi_{Peirce2}$}
\noLine
\UnaryInfC{$((D\imply C)\imply D)\imply D$}
\AxiomC{$\Pi_{Peirce1}$}
\noLine
\UnaryInfC{$((A\imply\xi)\imply A)\imply A$}
\AxiomC{$[(((A\imply \xi)\imply A)\imply A)\imply \xi]^{1}$}
\BinaryInfC{$\xi$}
\BinaryInfC{$C$}
\RightLabel{1}
\UnaryInfC{$((((A\imply \xi)\imply A)\imply A)\imply \xi)\imply C$}
\end{prooftree}
}
\caption{Proof of the formula in purely implicational {\em classical} logic}\label{outraprova}
\end{figure}

\section{No bounds for occurrence  assumptions in \mil}
\label{limits}

In this section we prove that for each $n$ there is a formula $\varphi_{n}$, such that, any normal proof of $\varphi_{n}$ has at least $2^{n}$ occurrence assumptions of the same formula, that are all of them discharged in only one introduction rule. 
The following proposition~\ref{upper} shows that $2^n$ is an upper bound by showing the normal proof that uses $2^{n}$ assumptions for proving $\varphi_{n}$. Theorem~\ref{lower} shows that there is no normal proof for any of the $\varphi_{n}$, in \mil, with less than $2^{n}$ assumptions discharged. In the sequel we define $\varphi_{n}$. As it was already said in section~\ref{sec:formula}, $\varphi_{n}$ raises from an iteration process derived from the previous examples. 

\begin{definition}
Let $\chi[X,Y]=(((X\imply Y)\imply X)\imply X)\imply Y$. Using $\chi[X,Y]$ we define recursively a family of formulas. Consider the propositional letters $D_i$, $C_i$, $i>0$. Let $\xi_i$, $i>0$, be the formula recursively defined as: 
\begin{eqnarray}
\xi_{1} & = & \chi[D_1,C] \\
\xi_{i+1} & = & \chi[D_{i+1},\xi_{i}] 
\end{eqnarray}
Using this family of formulas we define the formula $\varphi_{n}$, $n>0$, such that, for any $i\ge 0$:
\[
\varphi_{i+1} = \xi_{i+1}\imply C  
\]
\end{definition}

We can observe that $\varphi_1=\xi_1\imply C$ can be proved by using proof $\Sigma$, replacing $\xi$ for $C$ and $A$ for $D_1$, and applying an $\imply$-introduction as the last rule. The obtained proof has 2 occurrence assumptions of the formula $\xi_1$. The proof of $\varphi_2$ is the proof shown in figure~\ref{prova}, replacing $\xi$ by $\xi_1$, $A$ by $D_2$ and $D$ by  $D_1$, resulting in the proof showed in figure~\ref{provaphi2}. Note that the discharged formula $(((D_{2}\imply \xi_{1})\imply D_{2})\imply D_{2})\imply \xi_{1}$ is broke in two lines for reasons of space economy. 
The following lemma will be used in the proof of proposition~\ref{upper}.

\begin{lemma}\label{subs}
In the formula $\xi_i$, $i>0$, if we proceed in a simultaneous substitution, replacing $C$ by $\xi_1$, and for each $k>0$, replacing $D_{k}$ by $D_{k+1}$, the resulting formula is $\chi[D_{i+1},\xi_i]$. 
\end{lemma}

\begin{proof}
This lemma is proved by induction on $i$. For $\xi_1$ we observe that replacing $C$ by $\xi_1$ and $D_1$ by $D_2$ in $\xi_1$, the resulting formula is $\chi[D_2,\xi_1]$. Assuming that for $i>0$, replacing of $C$ by $\xi_1$ and, for each $k=1,i$, simultaneously replacing  $D_i$ by $D_{i+1}$ in $\xi_i$,  yields $\chi[D_{i+1},\xi_i]$. Observing that $\xi_{i+1}=\chi[D_{i+1},\xi_i]$ and by inductive hypothesis, simultaneous replacing of $C$ by $\xi_1$ and $D_k$ by $D_{k+1}$ in $\xi_{i}$, $k=1,i$, yields $\xi_{i+1}$. As $D_{i+1}$ does not occur in $\xi_i$, finally replacing $D_{i+1}$ by $D_{i+2}$ in $\xi_{i+1}=\chi[D_{i+1},\xi_{i+1}]$ yields $\chi[D_{i+2},\xi_{i+1}]$. This proves the inductive step. 
\end{proof}

Another observation is that substitutions as the above shown in the lemma, if applied in a derivation $\Pi$ in \mil, do imply that the resulting tree is a valid derivation too. This fact is justified by observing that the replacements are always on atomic formulas and the rules of \mil do not have provisos to be unsatisfied as consequence of these replacements. Thus,we have the following fact.

\begin{fact}\label{fato}
If $\Pi$ is a derivation of $\alpha$ from $\gamma_1,\ldots,\gamma_l$ and a substitution $\mathcal{S}$ (of atomic formulas only) is applied to $\Pi$ then $\mathcal{S}(\Pi)$ is a derivation of $\mathcal{S}(\alpha)$ from $\mathcal{S}(\gamma_1),\ldots,\mathcal{S}(\gamma_l)$. Besides that, if $\Pi$ is normal then $\mathcal{S}(\Pi)$ is normal too. 
\end{fact} 

\begin{landscape}
\begin{figure}[h]
{\tiny
\begin{prooftree}
\def\defaultHypSeparation{\hskip 0mm}
\insertBetweenHyps{\hskip -4cm}
\AxiomC{$[D_{1}]^{3}$}
\UnaryInfC{$(((D_{1}\imply C)\imply D_{1})\imply D_{1})$}
\AxiomC{$[D_{2}]^{1}$}
\UnaryInfC{$((D_{2}\imply \xi_{1})\imply D_{2})\imply D_{2}$}
\AxiomC{$[(((D_{2}\imply \xi_{1})\imply D_{2})$}
\noLine
\UnaryInfC{$\imply D_{2})\imply \xi_{1}]^{5}$}
\BinaryInfC{$\xi_{1}$}
\RightLabel{1}
\UnaryInfC{$D_{2}\imply\xi_{1}$}
\AxiomC{$[(D_{2}\imply\xi_{1})\imply D_{2}]^{2}$}
\BinaryInfC{$D_{2}$}
\RightLabel{2}
\UnaryInfC{$((D_{2}\imply\xi_{1})\imply D_{2})\imply D_{2}$}
\AxiomC{$[(((D_{2}\imply \xi_{1})\imply D_{2})$}
\noLine
\UnaryInfC{$\imply D_{2})\imply \xi_{1}]^{5}$}
\BinaryInfC{$\xi_{1}$}
\BinaryInfC{$C$}
\RightLabel{$3$}
\UnaryInfC{$D_{1}\imply C$}
\AxiomC{$[(D_{1}\imply C)\imply D_{1}]^{4}$}
\BinaryInfC{$D_{1}$}
\RightLabel{4}
\UnaryInfC{$((D_{1}\imply C)\imply D_{1})\imply D_{1}$}
\AxiomC{$[(((D_{2}\imply \xi_{1})\imply D_{2})$}
\noLine
\UnaryInfC{$\imply D_{2})\imply \xi_{1}]^{5}$}
\noLine
\UnaryInfC{$\Sigma$}
\noLine
\UnaryInfC{$\xi_{1}$}
\BinaryInfC{$C$}
\RightLabel{5}
\UnaryInfC{$((((D_{2}\imply \xi_{1})\imply D_{2})\imply D_{2})\imply \xi_{1})\imply C$}
\noLine
\UnaryInfC{}
\end{prooftree}
\begin{prooftree}
\AxiomC{$(((D_{2}\imply \xi_{1})\imply D_{2})\imply D_{2})\imply \xi_{1}$}
\noLine
\UnaryInfC{$\Sigma$}
\noLine
\UnaryInfC{$\xi_{1}$}
\AxiomC{is}
\AxiomC{$[D_{2}]$}
\UnaryInfC{$[((D_{2}\imply \xi_{1})\imply D_{2})\imply D_{2}]^{1}$}
\AxiomC{$[(((D_{2}\imply \xi_{1})\imply D_{2})\imply D_{2})\imply \xi_{1}]^{5}$}
\BinaryInfC{$\xi_{1}$}
\RightLabel{1}
\UnaryInfC{$D_{2}\imply\xi_{1}$}
\AxiomC{$[(D_{2}\imply\xi_{1})\imply D_{2}]^{2}$}
\BinaryInfC{$D_{2}$}
\RightLabel{2}
\UnaryInfC{$((D_{2}\imply\xi_{1})\imply D_{2})\imply D_{2}$}
\AxiomC{$[(((D_{2}\imply \xi_{1})\imply D_{2})$}
\noLine
\UnaryInfC{$\imply D_{2})\imply \xi_{1}]^{5}$}
\BinaryInfC{$\xi_{1}$}
\noLine
\TrinaryInfC{}
\end{prooftree}
}
\caption{Proof of $\varphi_{2}$ in \mil}\label{provaphi2}
\end{figure}
\end{landscape}

As $\varphi_{1}$ has two ($2^{1}$) occurrences of the same assumption and $\varphi_{2}$ has four ($2^{2}$) occurrences of the same assumptions, we have the following result.

\begin{proposition}\label{upper}
For any $n>0$, there is a normal proof of $\varphi_n$ having $2^{n}$ occurrences of the same assumptions, that are  discharged by the last rule of the proof.  
\end{proposition}

\begin{proof}
The proof proceeds by induction. The basis $n=1$ is the proof shown inside figure~\ref{provaphi2}. Assuming that $\varphi_{i}$, $i>0$ has a normal proof $\Pi_{\varphi_i}$ having $2^i$ occurrences of $\xi_{i}$ discharged by its last inference rule. Thus, we have a normal derivation $\Pi$ of C from $2^i$ occurrences of $\xi_{i}$, remembering that $\varphi_{i}=\xi_{i}\imply C$. We argue that if we simultaneously replace $C$ by $\xi_{1}$, and for each $k=1,i$, replace $D_{k}$ by $D_{k+1}$, we will have, by lemma~\ref{subs} and fact~\ref{fato}, a normal derivation of $\xi_{1}$ from $2^i$ occurrences of $\chi[D_{i+1},\xi_{i}]$. Let us call this derivation $\Pi^{\star}$. The following derivation (see figure~\ref{ind-hyp}) is a derivation of $C$ from  
$((((D_{i+1}\imply \xi_{i})\imply D_{i+1})\imply D_{i+1})\imply \xi_{i})\imply C$, i.e., it is a derivation of $C$ from 
$\xi_{i+1}$, and hence, by an $\imply$-introduction of we have a normal derivation of $\varphi_{i+1}$ discharging $2^{i}+2^{i}=2^{i+1}$ assumptions of the formula $\xi_{i+1}$ 
\begin{figure}[h]
{\tiny
\begin{prooftree}
\def\defaultHypSeparation{\hskip 0mm}
\insertBetweenHyps{\hskip -1cm}
\AxiomC{$[D_{1}]^{3}$}
\UnaryInfC{$(((D_{1}\imply C)\imply D_{1})\imply D_{1})$}
\AxiomC{$[(((D_{i+1}\imply \xi_{i})\imply D_{i+1})\imply D_{i+1})\imply \xi_{i}]^{5}$}
\noLine
\UnaryInfC{$\Pi^{\star}$}
\noLine
\UnaryInfC{$\xi_{1}$}
\BinaryInfC{$C$}
\RightLabel{$3$}
\UnaryInfC{$D_{1}\imply C$}
\AxiomC{$[(D_{1}\imply C)\imply D_{1}]^{4}$}
\BinaryInfC{$D_{1}$}
\RightLabel{$4$}
\UnaryInfC{$((D_{1}\imply C)\imply D_{1})\imply D_{1}$}
\AxiomC{$[(((D_{i+1}\imply \xi_{i})\imply D_{i+1})\imply D_{i+1})\imply \xi_{i}]^{5}$}
\noLine
\UnaryInfC{$\Pi^{\star}$}
\noLine
\UnaryInfC{$\xi_{1}$}
\BinaryInfC{$C$}
\RightLabel{5}
\UnaryInfC{$((((D_{i+1}\imply \xi_{i})\imply D_{i+1})\imply D_{i+1})\imply \xi_{i})\imply C$}
\end{prooftree}
}
\caption{Proof of $\varphi_{i+1}$ in \mil with $2^{i+1}$ discharged assumptions of $\xi_{i+1}$}\label{ind-hyp}
\end{figure}
\end{proof}
\begin{center}
Q.E.D.
\end{center}

The following proposition provides $2^{i}$ as the lower bound for number of assumption occurrences of a sole formula in proving $\varphi_i$ by means of normal proofs in \mil.

\begin{theorem}\label{lower}
Any normal proof of $\varphi_i$ in \mil has at least $2^i$ assumption occurrences of $\xi_i$.
\end{theorem}

\begin{proof}
We prove that for any $i$, there is no normal proof of $\varphi_i$ with less than $2^{i}$ assumption occurrences of $\xi_i$.
We first observe that $\varphi_1$, i.e., $((((D_1\imply C)\imply D_1)\imply D_1)\imply C)\imply C$  is not provable with 
only one occurrence of $\xi_1=(((D_1\imply C)\imply D_1)\imply D_1)\imply C)$. If this was the case we would have
that $((D_1\imply C)\imply D_1)\imply D_1$ is provable in \mil, and this cannot be since this formula is only classically valid.
A Kripke model with two worlds such that in the first world neither $C$ nor $D_1$ holds and in second $D_1$ holds but not $C$ falsifies $(((D_1\imply C)\imply D_1)\imply D_1)\imply C$.

Consider that there are normal proofs of $\varphi_i$ with less than $2^i$ assumption occurrences of $\xi_i$. So there is the least $k$ ($k>0$), such that, $\varphi_k$ has a normal proof with less than $2^k$ assumption occurrences of $\xi_k$. Let $\Sigma_k$ be such proof. Since $\varphi_k=\xi_k\imply C$, this proof is as follows. We remember that $\xi_k$ is the only open assumption in $\Sigma_k$. 
\begin{prooftree}
\AxiomC{$[\xi_k]^l$}
\noLine
\UnaryInfC{$\Sigma_k$}
\noLine
\UnaryInfC{$C$}
\RightLabel{$l$}
\UnaryInfC{$\xi_k\imply C$}
\end{prooftree}
Since $\xi_k=\chi[D_k,\xi_{k-1}]=(((D_k\imply \xi_{k-1})\imply D_k)\imply D_k)\imply \xi_{k-1}$, it has to be major premise of an $\imply$-elim rule. If this is not the case then $\xi_k$ is minor premise of a $\imply$-elim rule having a major premise of the form $\xi_k\imply\beta$. This formula on its turn has to be sub-formula of the open assumption of this branch, for the derivation is normal and $\xi_k\imply\beta$ can be only conclusion of an application of an $\imply$-elim rule. Since the only open assumption in $\Sigma_k$ is $\xi_k$ itself, the case of $\xi_k$ as minor premise is not possible. Thus, as $\xi_k$ is major premise, $\Sigma_k$ is of the following form, remembering how is $\xi_k$, showed  in the first line of this paragraph.
\begin{prooftree}
\AxiomC{$\Sigma^{\prime}$}
\noLine
\UnaryInfC{$(((D_k\imply \xi_{k-1})\imply D_k)\imply D_k)$}
\AxiomC{$[(((D_k\imply \xi_{k-1})\imply D_k)\imply D_k)\imply \xi_{k-1}]^l$}
\BinaryInfC{$\xi_{k-1}$}
\noLine
\UnaryInfC{$\Sigma_k$}
\noLine
\UnaryInfC{$C$}
\RightLabel{$l$}
\UnaryInfC{$\xi_k\imply C$}
\end{prooftree}
Note that $\Sigma^{\prime}$ is a sub-derivation of $\Sigma_k$ and it may have $\xi_k$ as open assumption too, but this is not necessary. If we remove every sub-derivation like $\Sigma^{\prime}$ from $\Sigma_k$ we end up with a proof as following:
\begin{prooftree}
\AxiomC{$[\xi_{k-1}]^l$}
\noLine
\UnaryInfC{$\Sigma_{k-1}$}
\noLine
\UnaryInfC{$C$}
\RightLabel{$l$}
\UnaryInfC{$\xi_{k-1}\imply C$}
\end{prooftree}
The proof above is a proof of $\varphi_{k-1}$ with less than $2^{k-1}$ assumption occurrences of $\xi_{k-1}$ discharged by the last rule. This contradicts the fact that $k$ is the least number holding this property. 
\end{proof}
\begin{center}
Q.E.D.
\end{center}

\section{Counter-model construction in Sequent Calculus for \mil}
\label{counter-model}
In this section we show how it is easy, from the computational aspect, to produce \mil counter-models from the following (incomplete) sequent calculus for \mil.
\begin{prooftree}
\AxiomC{}
\RightLabel{Axiom}
\UnaryInfC{\Seq{\Xi,p}{p,[\Delta]}}
\end{prooftree}

\begin{prooftree}
\AxiomC{\Seq{\Xi,\gamma_1}{\gamma_2,[\Delta]}}
\RightLabel{$\imply$-right}
\UnaryInfC{\Seq{\Xi}{\gamma_1\imply\gamma_2,[\Delta]}}
\end{prooftree}

\begin{prooftree}
\AxiomC{\Seq{\Xi}{\alpha,[\gamma,\Delta]}}
\AxiomC{\Seq{\Xi,\beta}{\gamma}}
\RightLabel{$\imply$-left}
\BinaryInfC{\Seq{\Xi,\alpha\imply\beta}{\gamma,[\Delta]}}
\end{prooftree}

The formulas in the right-hand side of the sequent and between the brackets are used only for counter-model construction. The main idea is that a sequent of the form $\Seq{\Xi,p}{q,[\Delta]}$ having all members of $\Xi\cup\Delta$ as propositional letters and $\{\Xi,p\}\cap\{q,\Delta\}=\emptyset$ is falsified in a Kripke model with only one world. From this case and using the reversible (if any premise is not valid the conclusion is too) rules of the system it is possible to build a polynomially sized Kripke model for the conclusion of the tree. Remember that in this case we do not have a proof. As already said, this system is incomplete, for it is unable to prove any of the formulas belonging to the class we presented here. The mentioned formulas are only provable if the $\imply$-left rule used is the following, instead of the above shown. 

\begin{prooftree}
\AxiomC{\Seq{\Xi,\alpha\imply\beta}{\alpha,[\gamma,\Delta]}}
\AxiomC{\Seq{\Xi,\alpha\imply\beta,\beta}{\gamma}}
\RightLabel{$\imply$-left}
\BinaryInfC{\Seq{\Xi,\alpha\imply\beta}{\gamma,[\Delta]}}
\end{prooftree}

In this case a counter-model generation is not so obvious, in fact we did not obtained one. If there were a bound on the use of repeated formulas, we could have used both versions of the $\imply$-rule for a counter-model generation. Of course, for every formula there is a bound, for example the formula $((((A\imply B)\imply A)\imply A)\imply B)\imply B$ the bound is 2. What we have shown  is that there is no fixed bound for every formula. In fact, if such a fixed bound existed we would have that every \mil formula would have a polynomially sized search-space to find either proofs or counter-models.  
  
\section{Conclusion}

Our contribution is in the context that \mil is the hardest and most representative propositional logic to define efficient proof-procedures. We show an example alerting for the fact that allowing unlimited use of assumptions is worth for any complete proof-procedure. This example runs in \mil. We are not aware of a similar example for classical logic. In this case classical propositional logic would be more efficient than \mil if such example does not existed. Propositional logic complexity has a lot of conjectures, starting with the relations between the main complexity classes. This article has the sole purpose of providing an example where the exponential grow of proofs has nothing to do with disjunction and combinatorial principles like the Pigeon-Hole\footnote{The pigeon-hole principle was used to provide a super-polynomial lower bound for Robinson's (propositional) Resolution}. We provided such example. 

Developers of theorem provers have to be aware of many aspects of the logic in order to design a efficient system. A system that saves memory and it is fast. Of course, dealing with PSPACE-complete problems is not a so easy task. Any information that can guide the designer is of help. Knowing that the number of copies of a formulas in a proof cab be a ``bottleneck''  for saving memory, an obvious solution would be the use of references instead of copies when representing proofs. The number of references is exponential, but references to formulas are smaller than formulas in most of the cases. This approach points out to the use of graphs (digraphs in fact) for representing proofs. There are a lot of developments done in this direction reported in the literature. Most of them are more semantically than implementation driven. Proof-nets (see \cite{Girard}) represents an approach that defends the use of graphs as the most adequate representation for proofs. We agree with that, but we add a practical motivation for considering digraphs instead of trees for representing proofs (see \cite{EPTCS}) 

\section{Acknowledgments} 

The author would like to thank prof. Gilles Dowek for hearing and reading the initial ideas presented here and providing very good suggestions. We want to thank Jefferson dos Santos for reading, pointing unclear explanations and helping improving the text. 

\bibliographystyle{plain}
\bibliography{example-more-twice}

\end{document}